\documentclass[runningheads]{llncs}

\usepackage{url}

\usepackage[margin=1in]{geometry}
\usepackage{graphicx}
\usepackage{textgreek}
\usepackage[font=small,labelfont=bf]{caption}
\usepackage{amsmath}
\usepackage{booktabs}
\usepackage{booktabs}
\usepackage{amssymb}
\usepackage{graphicx}
\usepackage{algpseudocode}
\usepackage{algorithm}
\usepackage{enumerate}
\usepackage{subcaption}
\usepackage{gensymb}
\usepackage{enumitem}
\usepackage[most]{tcolorbox}
\usepackage{mathtools}
\usepackage{xcolor}

\long\def\/*#1*/{}

\definecolor{darkpurple}{RGB}{68, 46, 72}
\definecolor{darkred}{RGB}{140,0,0}
\definecolor{darkgreen}{RGB}{0,100,0}
\newtheorem{observation}[definition]{Observation}

\newtheorem{Theorem}[definition]{Theorem}
\newtheorem{Lemma}[definition]{Lemma}

\author{Michael Amir, Noa Agmon and Alfred M. Bruckstein}

\begin{document}
\mainmatter              
\title{A Discrete Model of Collective Marching on Rings}
%
%
\author{Michael Amir\inst{1} \and Noa Agmon\inst{2} \and Alfred M. Bruckstein\inst{1}}

%
\tocauthor{Ivar Ekeland, Roger Temam, Jeffrey Dean, David Grove,
Craig Chambers, Kim B. Bruce, and Elisa Bertino}
\institute{Technion - Israel Institute of Technology,\\
\email{ammicha3@technion.ac.il},
\email{freddy@technion.ac.il}
\and
 Bar-Ilan University, Department of Computer Science\\
\email{agmon@cs.biu.ac.il}}

\maketitle              

\begin{abstract}
We study the collective motion of autonomous mobile agents on a ringlike environment. The agents' dynamics is inspired by known laboratory experiments on the dynamics of locust swarms. In these experiments, locusts placed at arbitrary locations and initial orientations on a ring-shaped arena are observed to eventually all march in the same direction. In this work we ask whether, and how fast, a similar phenomenon occurs in a stochastic swarm of simple agents whose goal is to maintain the same direction of motion for as long as possible. The agents are randomly initiated as marching either clockwise or counterclockwise on a wide ring-shaped region, which we model as $k$ ``narrow'' concentric tracks on a cylinder. Collisions cause agents to change their direction of motion. To avoid this, agents may decide to switch tracks so as to merge with platoons of agents marching in their direction.

We prove that such agents must eventually converge to a local consensus about their direction of motion, meaning that all agents on each narrow track must eventually march in the same direction. We give asymptotic bounds for the expected amount of time it takes for such convergence or ``stabilization'' to occur, which depends on the number of agents, the length of the tracks, and the number of tracks. We show that when agents also have a small probability of ``erratic'', random track-jumping behaviour, a global consensus on the direction of motion across all tracks will eventually be reached. Finally, we verify our theoretical findings in numerical simulations.

\end{abstract}

\section{Introduction}

Birds, locusts, human crowds and swarm-robotic systems exhibit interesting collective motion patterns. The underlying autonomous agentic behaviours from which these patterns emerge have attracted a great deal of academic interest over the last several decades  \cite{altshuler2018introduction,Ariel2015locustmarching,Fridman2010crowdbehaviour,ants3}. In particular, the formal analysis of models of swarm dynamics has led to varied and deep mathematical results \cite{ants1,chate2008modelingvicsek,czirok1999collective,ried2019modellinglocusts}. Rigorous mathematical results are necessary for understanding swarms and for designing predictable and provably effective swarm-robotic systems. However, multi-agent swarms have a uniquely complex and ``mesoscopic'' nature \cite{chazelle2018towardmesoscopic}, and relatively few standard techniques for the analysis of such systems have been established. Consequently, the analysis of new models of swarm dynamics is important for advancing our understanding of the subject.

In this work, we study the dynamics of ``locust-like'' agents moving on a  discrete ringlike surface. The model we study is inspired by the following well-documented experiment \cite{Amichay2016locusttopographyarena}: place many locusts on a ringlike arena at random positions and orientations. They start to move around and bump into the arena's walls and into each other, and as they do so, remarkably, over time, they begin to collectively march in the same direction--either clockwise or counterclockwise (see Figure \ref{reallocustsfigure}). Inspired by observing these experiments, we asked the following question: what are simple and reasonable myopic rules of behaviour that might lead to this phenomenon? Our goal is to study this question from an \textit{algorithmic} perspective, by considering a model of discretized mobile agents that act upon a local algorithm. As with much of the literature on swarm dynamics \cite{chazelle2012natural,ants1,probabilisticpursuits}, our goal is not to study an exact mathematical model of locusts in particular (the precise mechanisms underlying locusts' behaviours are very complex and subject to intense ongoing research, e.g. \cite{Amichay2016locusttopographyarena,Ariel2015locustmarching}), but to study the kinds of algorithmic local interactions that lead to collective marching and related phenomena. The resulting model is idealized and simple to describe, but the patterns of motion that emerge while the locusts progress towards a ``stabilized'' state of collective marching are surprisingly complex.

\begin{figure}[!ht]
\centerline{\includegraphics[scale=.2]{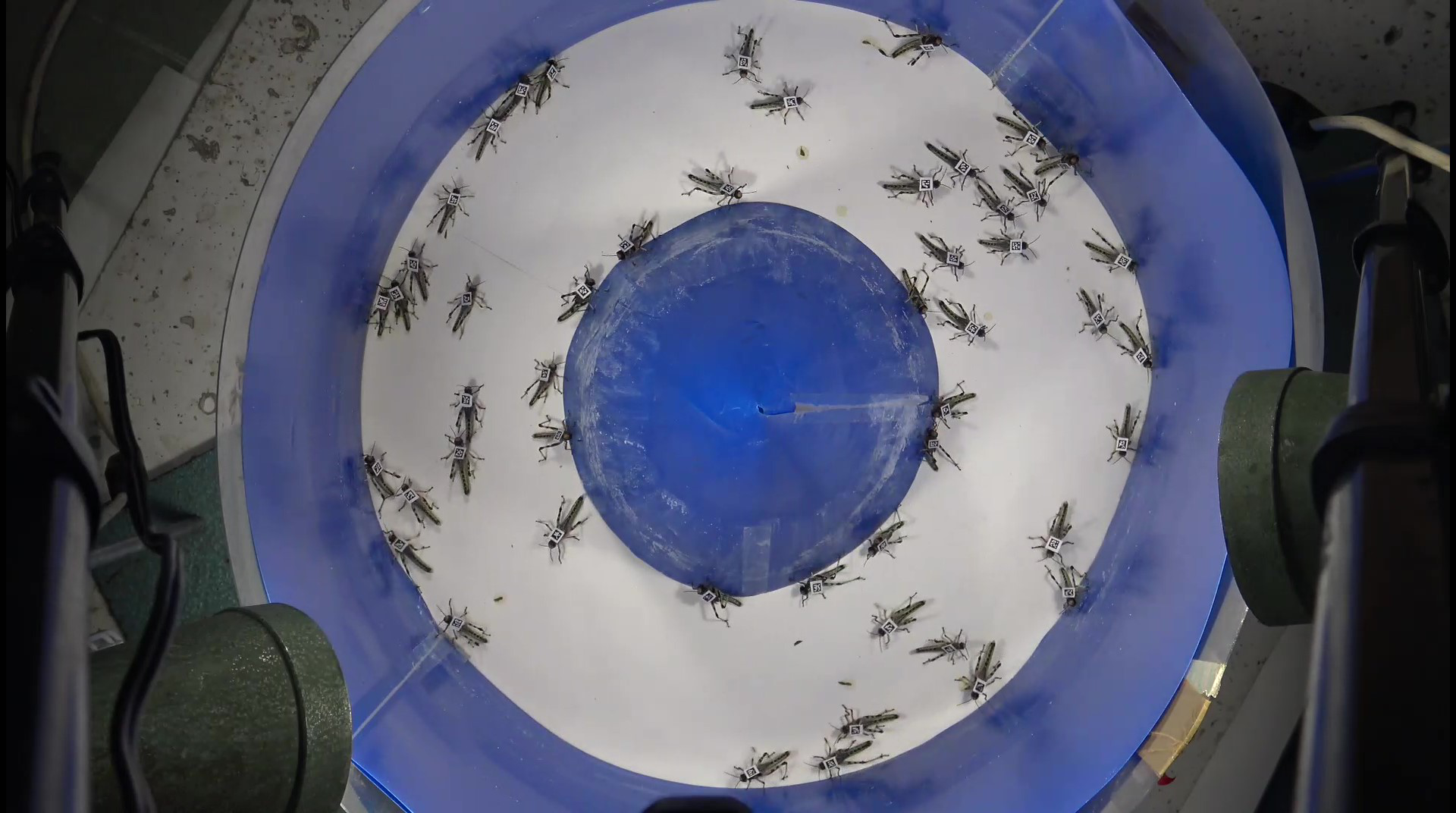}}
\caption{The collective clockwise marching of locusts in a ring arena (image by Amir Ayali).}
\label{reallocustsfigure}
\vspace{-6mm}
\end{figure}

The starting point for this work is the following postulated  ``rationalization'' of what a locust-like agent wants to do: it wants to keep moving in the same direction of motion (clockwise or counterclockwise) for as long as possible. We can therefore consider a model of locust-like agents that never change their heading unless they collide, heads-on, with agents marching in the opposite direction, and are forced to do so due to the pressure which is exerted on them. When possible, these agents prefer to \textit{bypass} agents that are headed towards them, rather than collide with those agents. This is done by changing lanes:  moving in an orthogonal manner between concentric narrow \textit{tracks} which partition the ringlike arena. The formal description of this ``rationalized'' model is given in Section \ref{modelsection}, and will be our subject of study.

\subsubsection{Contribution.} We describe and study a stochastic model of locust-like agents in a discretized ringlike arena which is modelled as multiple \textit{tracks} that wrap around a cylinder. We show that our agents eventually reach a ``local consensus'' about the direction of marching, meaning that all agents on the same track will march in the same direction. We give asymptotic bounds for the amount of time this takes based on the number of agents and the physical dimensions of the arena. Because of the idealized ``precise'' nature of our model, a global consensus where \textit{all} locusts walk in the same direction is not guaranteed, since locusts in different tracks might never meet. However, we show that, when a small probability of ``erratic'', random behaviour is added to the model, such a global consensus must occur. We verify our claims via simulations and make some further empirical observations that may inspire future investigations into the model. 

Despite being simple to describe, analyzing the model proved tricky in several respects. Our analysis strategy is to show that the model repeatedly passes between two phases: one in which it is ``chaotic'', such that locusts are arbitrarily moving about, and one in which it is ``orderly'', such that all locusts are in a kind of dense deadlock situation and collisions are frequent. We derive our asymptotic bounds from studying the well-behaved phase while bounding the amount of time the locusts can spend in the chaotic phase.

\subsubsection{Related work.} The experiments inspiring our work are discussed in \cite{Ariel2015locustmarching,Amichay2016locusttopographyarena}. The mathematical modelling of the collective motion of natural organisms such as birds, locusts and ants, and the convergence of such systems of agents to stable formations, has been discussed in numerous works including \cite{czirok1999collective,ants3,ried2019modellinglocusts,shiraishi2015collectivelyapunovswarmdynamics}. 

The central focus of this work regards consensus: do the agents eventually converge to the same direction of motion, and how long does it take? Similar questions are often asked in the field of opinion dynamics. Mathematically, if the agents' direction of motion (clockwise or counterclockwise) is considered an ``opinion'', we can compare our model to models in this field. When there are no empty locations at all in the environment, our model is fairly close to the \textit{voter model} on a ring network with two distinct opinions, the main difference being that, unlike in the voter model, our agents' direction of motion determines which agents' opinions can influence them (an excellent survey on this topic is \cite{dong2018surveyvoteretc}). The comparison to the voter model breaks when we introduce empty locations and multiple ringlike tracks, at which point we must take into account the physical location of every agent when considering which agents can influence its opinion. Several works have explored models of opinion dynamics in a ring environment where the agents' physical location is taken into account \cite{chandra2017diffusionring,hegarty2016hegselmann}. Our model is distinct from these in several respects: first, in our model, an agents' internal state--its direction of motion--plays an active part in the algorithm that determines which locations an agent may move to. Second, we partition our ring topology into several narrow rings (``tracks'') that agents may switch between, and an agents' decision to switch tracks is influenced by the presence of platoons of agents moving in its direction in the track that it wants to switch to. In other words, we model agents that actively attempt to ``swarm'' together with agents moving in their direction of motion

Protocols for achieving consensus about a value, location or the collective direction of motion have also been investigated in swarm robotics and distributed algorithms  \cite{barel2017come,cortes2008distributedconsensus,manor2018chase,olfati2007consensus}. However, in this work, we are not searching for a protocol that is designed to efficiently bring about consensus; we are investigating a protocol that is inspired by natural phenomena and want to see \textit{whether} it leads to consensus and how long this takes on average.

Broadly speaking, some mathematical similarities may be drawn between our model and interacting particle systems such as the simple exclusion process, which have been used to understand biological transport and traffic phenomena \cite{chou2011biologicaltasep}. Such particle systems have been studied on rings \cite{kriecherbauer2010pedestrian}. In these discrete models, as in our model, agents possess a physical dimension, which constrains the locations they might move to in their environment. These are not typically multi-agent models where agents have an internal state (such as a persistent direction of motion), but rather models of particle motion and diffusion, and the research focus is quite different; the main point of similarity to our model is in the way that a given discrete location can only be occupied by a single agent, and in the random occurrence of ``traffic shocks'' wherein agents line up one after the other and are prevented from moving for a long time. 

\section{Model and definitions}
\label{modelsection}

We postulate a locust-inspired model of marching in a wide ringlike arena which is divided into narrow concentric rings. For simplicity, we map the arena to the surface of a discretized cylinder of height $k$ partitioned into $k$ narrow rings of length $n$, which are called \textit{tracks}. For example, the environment of Figure \ref{locuststepsfigure} corresponds to $k=3, n=8$ ($3$ tracks of length $8$). The coordinate $(x,y)$ refers to the $x$th location on the $y$th track (which can also be seen as the $x$th location of a ring of length $n$ wrapped around the cylinder at height $y$). Since we are on a cylinder, we have that $\forall x, (x+n, y) \equiv (x,y)$. 

A swarm of $m$ identical agents, or ``locusts,'' which we label $A_1, \ldots, A_m$, are dispersed at arbitrary locations on the cylinder and move autonomously at discrete time steps $t = 0, 1, \ldots$. A given location $(x,y)$ on the cylinder can contain at most one locust.  Each locust $A_i$ is initiated with either a ``clockwise'' or ``counterclockwise'' \textit{heading}, which determines their present direction of motion. We define $b(A_i) = 1$ when $A_i$ has clockwise heading, and $b(A_i) = -1$ when $A_i$ has counterclockwise heading. 

The locusts move synchronously at discrete time steps $t = 0, 1, \ldots$. At every time step, locusts try to take a step in their direction of motion: if a locust $A$ is at $(x,y)$, it will attempt to move to $(x+b(A), y)$. A clockwise movement corresponds to adding $1$ to $x$, and a counterclockwise movement corresponds to subtracting $1$. The locusts have physical dimension, so if the location a locust attempts to move to already contains another locust at the beginning of the time step, the locust instead stays put. If $A_i$ and $A_j$ are both attempting to move to the same location, one of them is chosen uniformly at random to move to the location and the other stays put.

Locusts that are adjacent exert pressure on each other to change their heading: if $A_i$ has a clockwise heading and $A_j$ has a counterclockwise heading, and they lie on the coordinates $(x,y)$ and $(x+1,y)$ respectively, then at the end of the current time step, one locust (chosen uniformly at random) will flip its heading to the other locust's heading. Such an event is called a \textbf{conflict} between $A_i$ and $A_j$. A conflict is ``won'' by the locust that successfully converts the other locust to their heading.

 Let $A$ be a locust at $(x,y)$. If the locust $A$ has clockwise heading, then the \textit{front} of $A$ is the first locust after $A$ in the clockwise direction, and the \textit{back} of $A$ is the first locust in the counterclockwise direction. The reverse is true when $A$ has counterclockwise heading. Formally, let $i > 0$ be the smallest positive integer such that  $(x+b(A) i, y)$ contains a locust, and let $j > 0$ be the smallest positive integer such that  $(x-b(A)j, y)$ contains a locust. The \textit{front} of $A$ is the locust in $(x+b(A) i, y)$ and the \textit{back} of $A$ is the locust in $(x-b(A) j, y)$. The locusts in the front and back of $A$ are denoted $A^{\rightarrow}$ and $A^{\leftarrow}$ respectively, and are called $A$'s \textit{neighbours}; these are the locusts that are directly in front of and behind $A$. Note that when a track has two or less locusts,  $A^{\rightarrow} = A^{\leftarrow}$. When a track has one locust, $i = j = n$ and so $A = A^{\rightarrow} = A^{\leftarrow}$.
 
 At any given time step, besides moving in the direction of their heading within their track, a locust $A$ at $(x,y)$ can switch tracks, moving vertically from $(x,y)$ to $(x,y+1)$ or $(x,y-1)$ (unless this would cause it to go above track $k$ or below track $1$). Such vertical movements occur \textit{after} the horizontal movements of locusts along the tracks, but on the same time step where those horizontal movements took place. Locusts are incentivized to move vertically when this enables them to avoid changing their heading (``inertia''). Specifically, $A$ may move to the location $E = (x,y \pm1)$ at time $t$ when:
 
 \begin{enumerate}
     \item At the beginning of time $t$, $A$ and $A^{\rightarrow}$ are not adjacent to each other and $b(A) \neq b(A^{\rightarrow})$.
     \item Once $A$ moves to $E$, the updated $A^{\leftarrow}$ and $A^{\rightarrow}$ in the new track will have heading $b(A)$.
     \item No locust will attempt to move horizontally to $E$ at time $t+1$.
 \end{enumerate}

Condition (1) states that there is an imminent conflict between $A$ and $A^{\rightarrow}$ which is bound to occur. Condition (2) guarantees that, by changing tracks to avoid this conflict, $A$ is not immediately advancing towards another collision; $A$'s new neighbours will have the same heading as $A$. Condition (3) guarantees that the location $A$ wants to move to on the new track isn't being contested by another locust already on that track. Together, these conditions mean that locusts only change tracks if this results in avoiding collisions and in ``swarming'' together with other locusts marching in the same direction of motion. If a locust cannot sense that all three conditions (1), (2) and (3) are fulfilled, it does not switch tracks.

Besides these conditions, we make no assumptions about  \textit{when} locusts move vertically. In other words, locusts do not always need to change tracks when they are allowed to by rules (1)-(3); they may do so arbitrarily, say with some probability $q$ or according to any internal scheduler or algorithm. We do not determine in any sense the times when locusts move tracks--but only determine the preconditions required for such movements; our results in the following sections remain true regardless. This makes our results general in the sense that they hold for many different track-switching ``swarming'' rules, so long as those rules do not break the conditions (1)-(3). 

Figure \ref{locuststepsfigure} illustrates one time step of the model, split into horizontal and vertical movement phases.

\begin{figure}[!ht]
\centerline{\includegraphics[scale=.3]{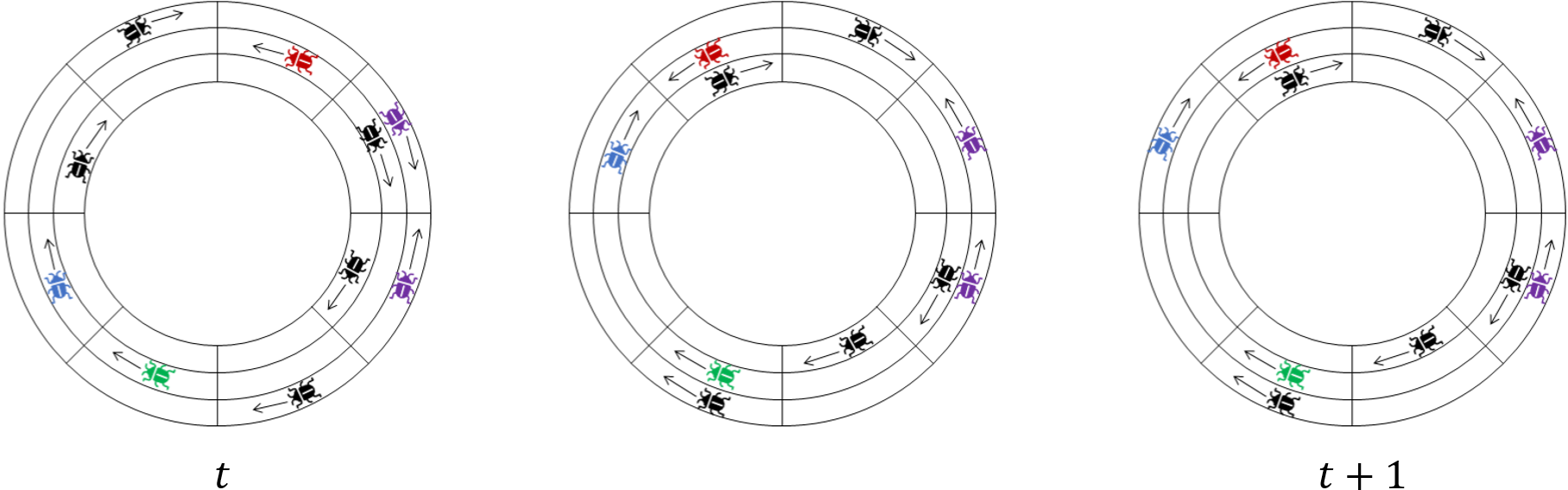}}
\caption{One step of the locust model with $k=3$, $n=8$ split into horizontal and vertical movements. The leftmost figure is the initial configuration at the beginning of the current time step $t$, the middle illustrates changes to the configuration after conflicts and horizontal movements, and the rightmost figure is the configuration at the beginning of time $t+1$ (or equivalently the end of time $t$) after vertical movements. The \textit{front} and \textit{back} of the \color{blue}{blue} \color{black} locust are the \color{darkred}{red}\color{black}\  and \color{darkgreen}{green}\color{black}\  locusts respectively. The \color{darkpurple}{purple}\color{black}\  locusts conflict with each other. Since conditions (1)-(3) are fulfilled, the \color{blue}{blue}\color{black}\  locust may switch tracks, and it does so.}
\label{locuststepsfigure}
\vspace{-6mm}
\end{figure}

In order to slightly simplify our analysis of the model, we assume that every track has at least $2$ locusts at all times, although our results remain true without this assumption.

Everywhere in this work, \textit{the beginning} of a time step refers to the configuration of the swarm at that time step before any locusts moved, and \textit{the end} of a time step refers to the configuration at that time step after all locust movements are complete. By default and unless stated otherwise, the words ``time step $t$'' refer to the beginning of that time step.

\section{Stabilization analysis}

We will mainly be interested studying the stability of the headings of the locusts over time. Does the model reach a point where the locusts stabilize stop changing their heading? If so, are their headings all identical? How long does it take?

In the case of a single track ($k = 1$), we shall see that the locusts all eventually stabilize with identical heading, and bound the expected time for this to happen in terms of $m$ and $n$. In the multi-track case, we shall see that the locusts stabilize and agree on a heading \textit{locally} (i.e., all locusts \textit{on the same track} eventually have identical heading and thereafter never change their heading), and bound the expected time to stabilization in terms of $m, n, k$. In the multi-track case, we show further that adding a small probability of ``erratic'' track-switching behaviour to the model induces \textit{global} consensus: all locusts across all tracks eventually have identical heading.

\subsection{Locusts on narrow ringlike arenas ($k=1$)}

We start by studying the case $k=1$, that is, we study a swarm of $m$ locusts marching on a single track of length $n$. Throughout this section, we assume this is the case, except in Definition \ref{segmentdefinition}, which is also used in later sections. 

For the rest of this section, let us call the swarm \textit{non-stable} at time $t$ if there are two locusts $A_i$ and $A_j$ such that $b(A_i) \neq b(A_j)$; otherwise, the swarm is \textit{stable}. A swarm which is stable at time $t$ remains stable thereafter. We wish to bound the number of time steps it takes for the system to become stable, which we denote $T_{stable}$. Our goal is to prove Theorem \ref{k1bounds}, which tells us that the expected time to stabilization grows quadratically in the number of locusts $m$, and linearly in the track length $n$.

\begin{Theorem}
\label{k1bounds}
For any configuration of $m$ locusts on a ring with a single track, $\mathbb{E}[T_{stable}] \leq  m^2 + 2(n - m)$. This bound is asymptotically tight: there are initial locust configurations for which $\mathbb{E}[T_{stable}] = \Omega(m^2 + n - m)$. 
\end{Theorem}

In particular, Theorem \ref{k1bounds} tells us that all locusts must have identical bias within finite expected time. This fact in isolation (without the time bounds in the statement of the theorem) is relatively straightforward to prove, by noting that the evolution of the locusts' headings and locations can be modelled as a finite Markov chain, and the only absorbing classes in this Markov chain are ones in which all locusts have the same heading (see  \cite{grinstead2012introduction}).

Next we define \textit{segments}: sets of consecutive locusts on the same track which all have the same heading.  This will allow us to partition the swarm into segments, such that every locust belongs to a unique segment (see Figure \ref{segmentfigure}). Although this section focuses on the case of a single track (and claims in this section are made under the assumption that there is only a single track), the definition is general, and we will use it in subsequent sections.

\begin{definition}
\label{segmentdefinition}

Let $A$ be a locust for which $b(A^{\leftarrow}) \neq b(A)$ at time $t$, and consider the sequence of locusts $B_0 = A$, $B_{i+1} = B_i^{\rightarrow}$. Let $B_q$ be the first locust in this sequence for which $b(B_q) \neq b(B_0)$. The set $\{B_0, B_1, \ldots B_{q-1}\}$ is called the \textbf{segment} of the locusts $B_0, \ldots B_{q-1}$  at time $t$. The locust $ B_{q-1}$ is called the \textbf{segment head}, and $A$ is called the \textbf{segment tail} of this segment.

\end{definition}

\begin{figure}[!ht]
\centerline{\includegraphics[scale=.3]{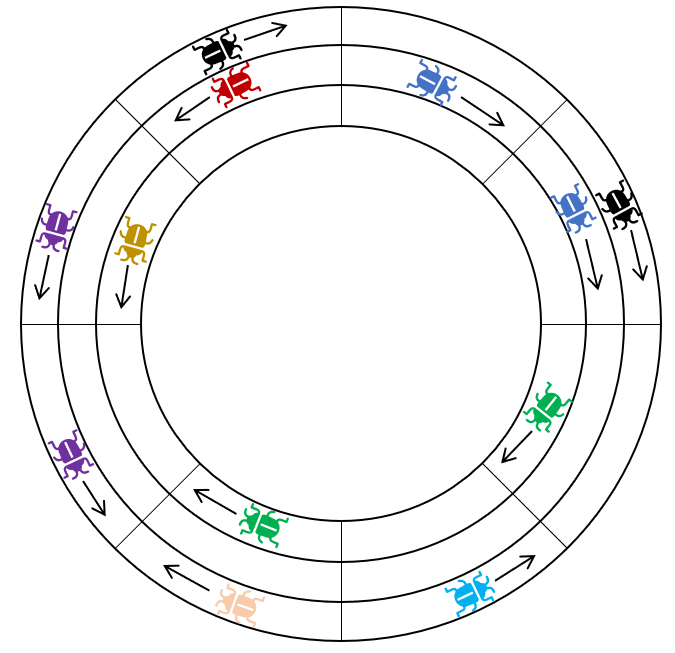}}
\caption{A locust configuration with $n=8, k=3$. Locusts are colored based on their segment.}
\label{segmentfigure}
\vspace{-6mm}
\end{figure}

Only locusts which are segment heads at the beginning of a time step can change their heading by the end of that time step. When the heads of two segments are adjacent to each other, the resulting conflict causes one to change its heading, leave its previous segment, and instead become part of the other segment. If the head of a segment is also the tail of a segment, the segment is eliminated when it changes heading. Two segments separated by a segment of opposite heading merge if the opposite-heading segment is eliminated, which decreases the number of segments by $2$. No other action by a locust can change the segments. Hence, the number of segments and segment tails can only decrease. 

Since our model is stochastic, different sequences of events may occur and result in different segments. However, by the above argument we can conclude that in any such sequence of events,  there must always exist at least one locust which remains a segment tail at all times $t < T_{stable}$ and never changes its heading (since at least one segment must exist as long as $t < T_{stable}$). Arbitrarily denote one such segment tail ``$A_W$''.  

\begin{definition}
\label{winningsegment}
The segment of $A_W$ at the beginning of time $t$ is called the \textbf{winning segment} at time $t$, and is denoted $SW(t)$. The head of $SW(t)$ is labelled $H_W(t)$. For convenience, if at time $t_0$ the swarm is stable (i.e. $t_0 \geq T_{stable}$), then we define $SW(t_0)$ as the set that contains all $m$ locusts.
\end{definition}

\begin{Lemma}
\label{changecountlemma}
The expected number of time steps $t < T_{stable}$ in which $|SW(t)|$ changes is bounded by $m^2$.
\end{Lemma}

\begin{proof}
Let $C_m$ denote the number of changes to the size of $SW(t)$ that occur before time $T_{stable}$. Note that $T_{stable}$ is the first time step where $|SW(t)| = m$. $|SW(t)|$ can only decrease, by $1$ locust at a time, if $H_W(t)$ conflicts with another locust and loses. $|SW(t)|$ can increase in several ways, for example when it merges with other segments. In particular, $|SW(t)|$ increases by at least $1$ whenever $H_W(t)$ conflicts with a locust and wins, which happens with probability at least $\frac{1}{2}$. Hence, whenever $SW(t)$ changes in size, it is more likely to grow than to shrink. We can bound $E[C_m]$ by comparing the growth of $|SW(t)|$ to a random walk with absorbing boundaries at $0$ and $m$: 

Consider a random walk on the integers which starts at $|SW(0)|$. At any time step $t$, the walker takes a step left with probability $\frac{1}{2}$, otherwise it takes a step right. If the walker reaches either $0$ or $m$,  the walk ends. Denote by $C^*_m$ the time it takes the walk to end. Using \textit{coupling} (cf. \cite{lindvall2002coupling}), we see that $\mathbb{E}[C_m] \leq \mathbb{E}[C^*_m | \text{the walker never reaches 0}]$, since per the previous paragraph, $|SW(t)|$ clearly grows at least as fast as the position of the random walker (note that $|SW(t)| > 0$ is always true, which is analogous to the walker never reaching $0$). 

Let us show how to bound $\mathbb{E}[C^*_m | \text{the walker never reaches 0}]$. Since the walk is memoryless, we can think of this quantity as the number of steps the random walker takes to get to $m$, assuming it must move right when it is at $0$, and assuming the step count restarts whenever it moves from $0$ to $1$. If we count the steps without resetting the count, we get that this is simply the expected number of steps it takes a random walker walled at $0$ to reach position $m$, which is at most $m^2$ (cf. \cite{aldous1995reversible}). Hence $\mathbb{E}[C^*_m | \text{the walker never reaches 0}] \leq m^2$. \qed\end{proof}

\begin{Lemma}
\label{nochangecount}
The expected number of time steps $t < T_{stable}$ in which $|SW(t)|$ does not change is bounded by $2(n - m)$.
\end{Lemma}

Lemma \ref{nochangecount} will require other lemmas, and some new definitions to prove.

\begin{definition}
\label{clockdistance}
Let $A$ and $B$ be two locusts or two locations which lie on the same track. The clockwise distance from $A$ to $B$ at time $t$ is the number of clockwise steps required to get from $A$'s location to $B$'s location, and is denoted $dist^{c}(A,B)$. The counterclockwise distance from $A$ to $B$ is denoted $dist^{cc}(A,B)$ and equals $dist^{c}(B,A)$.
\end{definition}

\textbf{For the rest of this section,} let us assume without loss of generality that the winning segment's tail $A_W$ has clockwise heading. Label the empty locations in the ring at time $t=0$  (i.e., the locations not containing locusts at time $t=0$) as $E_1, E_2, \ldots E_{n - m}$, sorted by their counterclockwise distance to $A_W$ at time $t=0$, such that $E_1$ minimizes $dist^{cc}(E_i, A_W)$, $E_2$ has the second smallest distance, and so on. We will treat these empty locations as having persistent identities: whenever a locust $A$ moves from its current location to $E_i$, we will instead say that $A$ and $E_i$ \textit{swapped}, and so $E_i$'s new location is $A$'s old location.

We say a location $E_i$ is \textit{inside} the segment $SW(t)$ at time $t$ if the two locusts which have the smallest clockwise and counterclockwise distance to $E_i$ respectively are both in $SW(t)$. Otherwise, we say that $E_i$ is \textit{outside} $SW(t)$. A locust or location $A$ is said to be \textit{between} $E_i$ and $E_j$, $j > i$, if $dist^{c}(E_i,A) < dist^{c}(E_i, E_j)$.

\begin{definition}All empty locations are initially \textbf{blocked}. A location $E_i$ becomes \textbf{unblocked} at time $t+1$ if all empty locations $E_j$ such that $j < i$ are unblocked at time $t$, and a locust from  $SW(t)$ swapped locations with $E_i$ at time $t$. Once a location becomes \textbf{unblocked}, it remains that way forever. 
\end{definition}

\begin{Lemma}
\label{unblockedoroutsidelemma}
There is some time step $t^* \leq n - m$ such that:

\begin{enumerate}
    \item Every blocked empty location $E$ is outside $SW(t^*)$ (if any exist)
    \item At least $t^*$ empty locations are unblocked.
\end{enumerate}

\end{Lemma}

\begin{proof}
If $E_1$ is outside $SW(0)$, then the same must be true for all other empty locations, so $t^* = 0$ and we are done. Otherwise, $E_1$ becomes unblocked at time $t=1$. If $E_i$ becomes unblocked at time $t$, then at time $t$, it cannot be adjacent to $E_{i+1}$, since the locust that swapped with $E_i$ in the previous time step is now between $E_i$ and $E_{i+1}$. By definition, there are no empty locations $E_j$ between $E_i$ and $E_{i+1}$. Consequently, if $E_{i+1}$ is inside $SW(t)$ at time $t$, it will swap with a locust of $SW(t)$ at time $t$, and become unblocked at time $t+1$. If $E_{i+1}$ is outside the segment at time $t$, it will become unblocked at the first time step $t'>t$ that begins with $E_{i+1}$ inside  $SW(t')$. Hence, if $E_i$ becomes unblocked at time $t$, then $E_{i+1}$ becomes unblocked at time $t+1$ or $E_{i+1}$ is outside $SW(t+1)$ at time $t+1$.

Let $t^*$ be the smallest time where there are no blocked empty locations inside $SW(t^*)$. By the above, at every time step $t \leq t^*$ an empty location becomes unblocked, hence there are at least $t^*$ unblocked empty locations at time $t^*$. Also, since there are $n - m$ empty locations, this implies $t^* \leq n - m$. \qed\end{proof}

\begin{Lemma}
\label{noswapwithunblockedlemma}
There is no time $t < T_{stable}$ where an unblocked location is clockwise-adjacent to $H_W(t)$ (i.e., there is no time $t$ where an  unblocked empty location $E$ is located one step clockwise from $H_W(t)$).
\end{Lemma}

\begin{proof}
First consider what happens when $E_1$ becomes unblocked: it swaps its location with a locust in $SW(t)$, and since $E_1$ is the clockwise-closest empty location to $A_W$, the entire counterclockwise path from $E_1$ to $A_W$ consist only of locusts from $SW(t)$. Hence $E_1$ will move counterclockwise at every time step, until it swaps with $A_W$. Once it swaps with $A_W$, $E_1$ will not swap with another locust at all times $t < T_{stable}$, since for that to occur we must have that $b(A_W^{\leftarrow})=b(A_W)$, which is impossible since by definition $A_W$ remains a segment tail until $t=T_{stable}$. $E_1$ does not swap with $H_W(t)$ while $E_1$ moves counterclockwise towards $A_W$ nor after $E_1$ and $A_W$ swap as long as the swarm is unstable, hence there is no time step $t < T_{stable}$ when $E_1$ is unblocked and swaps with $H_W(t)$.

Now consider $E_2$. $E_2$ becomes unblocked at least one time step after $E_1$, and there is at least one locust in $SW(t)$ which is between $E_1$ and $E_2$ at the time step  $E_1$ becomes unblocked (in particular, the locust in $SW(t)$ that swapped with $E_1$ must be between $E_1$ and $E_2$ at that time). Since $E_1$ subsequently moves towards $A_W$ at every time step until they swap, $E_2$ cannot  become adjacent to $E_1$ until they both swap with $A_W$. Hence the location one step counterclockwise to $E_2$ must always be a locust until $E_2$ swaps with $A_W$, meaning that similar to $E_1$, $E_2$ also moves counterclockwise towards $A_W$ at every time step after $E_2$ becomes unblocked until they swap locations. Consequently, just like $E_1$, there is no time step $t < T_{stable}$ when $E_2$ is unblocked and swaps with $H_W(t)$.

More generally, by a straightforward inductive argument, the exact same thing is true of $E_i$: once it becomes unblocked, it moves counterclockwise towards $A_W$ at every time step until it swaps with $A_W$. Thus, upon becoming unblocked, $E_i$ does not swap with $H_W(t)$ as long as $t < T_{stable}$.\qed\end{proof}

Using Lemmas \ref{unblockedoroutsidelemma} and \ref{noswapwithunblockedlemma}, let us prove Lemma \ref{nochangecount}. 

\begin{proof}
If, at the beginning of time step $t$, $H_W(t)$ is adjacent to a locust from a different segment, then $|SW(t)|$ will change at the end of this time step due to the locusts' conflict. Hence, to prove Lemma \ref{nochangecount}, it suffices to show that out of all the time steps before time $T_{stable}$, $H_W(t)$ is not adjacent to the head of a different segment in at most $2(n - m)$ different steps in expectation.

If all empty locations are unblocked at time $n - m$, then by Lemma \ref{noswapwithunblockedlemma}, $H_W(t)$ conflicts with the head of another segment at all times $t \geq n - m$. Therefore, $|SW(t)|$ will change at every time step $n - m < t < T_{stable}$, which is what we wanted to prove.

If all empty locations are not unblocked by time $n - m$, then by Lemma \ref{nochangecount}, there must be some time $t^* \leq n - m$ where at least $t^*$ empty locations are unblocked and all blocked empty locations are outside $SW(t^*)$. Let $E_j$ be the minimal-index blocked location which is outside $SW(t^*)$ at time $t^*$. Since there are no blocked empty locations inside $SW(t^*)$, all locations $E_i$ with $i < j$ are unblocked. Hence, $E_j$ will become unblocked as soon as it swaps with the head of the winning segment. Since (by the clockwise sorting order of $E_1, E_2, \ldots$) $E_{j+1}$ cannot swap with the winning segment head  before $E_j$ is unblocked, $E_{j+1}$ will also become unblocked after the first time step where it swaps the winning segment head. The same is true for $E_{j+2}, \ldots E_{n-m}$. Hence, every empty location that $H_W(t)$ swaps with after time $t^*$ becomes unblocked in the subsequent time step. By Lemma \ref{nochangecount}, the total swaps $H_W(t)$ could have made before time $T_{stable}$ is thus most $t^* + (n-m-j) \leq n - m$. Whenever an empty location is one step clockwise from $H_W(t)$, they will swap with probability at least $0.5$ (the swap is not guaranteed, since it is possible the location is also adjacent to the head of another segment, and hence a tiebreaker will occur in regards to which segment head occupies the empty location in the next time step). Consequently, the expected number of time steps $H_W(t)$ is not adjacent to the head of another segment is bounded by $2(n-m)$.\qed\end{proof}

The proof of Theorem \ref{k1bounds} now follows.

\begin{proof} Lemma \ref{nochangecount} tells us that before time $T_{stable}$, $|SW(t)|$ does not change in at most $2(n - m)$ time steps in expectation, whereas Lemma \ref{changecountlemma} tells us that the expected number of changes to $|SW(t)|$ before time $T_{stable}$ is at most $m^2$. Hence, for any configuration of $m$ locusts on a ring of track length $n$, $\mathbb{E}[T_{stable}] \leq m^2 + 2(n - m)$. 

Let us now show a locust configuration for which $\mathbb{E}[T_{stable}] = \Omega(m^2+n)$, so as to asymptotically match the upper bound we found.  Consider a ring with $k=1$, $m$  divisible by $2$, and an initial locust configuration where locusts are found at coordinates $(0,1), (1,1), \ldots (m/2,1)$ with clockwise heading and at $(-1,1), (-2,1), \ldots (-m/2-1,1)$ with counterclockwise heading, and the rest of the ring is empty. This is a ring with exactly two segments, each of size $m/2$. Since after every  conflict, the segment sizes are offset by $1$ in either direction, the expected number of conflicts between the heads of the segments that is necessary for stabilization is equal to the expected number of steps a random walk with absorbing boundaries at $m/2$ and $-m/2$ takes to end, which is $m^2/4$ (see \cite{epstein2012gamblersruin}). Since the heads of the segments start at distance $n-m$ from each other, it takes $\Omega(n-m)$ steps for them to reach each other. Hence the expected time for this ring to stabilize is $\Omega(m^2 + n - m)$.\qed\end{proof}

\subsection{Locusts on wide ringlike arenas ($k > 1$)}

Let us now investigate the case where $m$ locusts are marching on a cylinder of height $k > 1$ partitioned into $k$ tracks of length $n$. The first question we should ask is whether, just as in the case of the $k=1$ setting, there exists some time $T$ where all locusts have identical heading. The answer is ``not necessarily'': consider for example the case $k=2$ where on the $k=1$ track, all locusts march clockwise, and on the $k=2$ track, all locusts march counterclockwise. According to the track-switching conditions (Section \ref{modelsection}), no locust will ever switch tracks in this configuration, hence the locusts will perpetually have opposing headings. As we shall prove in this section, on the cylinder, swarms stabilize \textit{locally}--meaning that eventually, all locusts \textit{on the same track} have identical heading, but this heading may be different between tracks.

Let us say that the $y$th track is stable if all locusts whose location is $(\cdot, y)$ have identical heading. Note that once a track becomes stable, it remains this way forever, as by the model, the only locusts that may move into the track must have the same heading as its locusts.  Let $T_{stable}$ be the first time when every all the $k$ tracks are stable. Our goal will be to prove the following asymptotic bounds on $T_{stable}$:

\begin{Theorem}
\label{klargebounds}
$\mathbb{E}[T_{stable}] =  \mathcal{O}\large(\min(\log (k)n^2, mn + m^2)\large)$.\footnote{In an upcoming extended version of this work, we refine these bounds to $\mathcal{O}\large(\min(\log (k)n^2, mn)\large)$}
\end{Theorem}

The bound $\mathcal{O}(mn + m^2)$ is more accurate when $m$ is small ($m \ll \sqrt{\log (k)} n$), and the bound  $\mathcal{O}(\log (k)n^2)$ is more accurate when $m$ is large.

Recalling Definition \ref{segmentdefinition}, each locust in the system belongs to some segment. Each track has its own segments. Locusts leave and join segments due to conflicts, or when they pass from their current segment to a track on a different segment. In this section, we will treat segments as having persistent identities, similar to $SW$ in the previous section. We introduce the following notation:

\begin{definition}
\label{segmentidentitydefinition}
Let $S$ be a segment whose tail is $A$ at some time $t_0$. We define $S(t)$ to be the segment whose tail is $A$ at the beginning of time $t$. If $A$ is not a segment tail at time $t$, then we will say $S(t) =  \emptyset$ (this can happen once $A$ changes its heading or moves to another track, or due to another segment merging with $S(t)$ which might cause $b(A^{\leftarrow})$ to equal $b(A)$, thus making $A$ no longer the tail).

Furthermore, define $S_1$ to be the segment tail of $S$ and $S_{i+1} = S_i^{\rightarrow}$. 

\end{definition}

Let us give a few examples of the notation in Definition \ref{segmentidentitydefinition}. Suppose at time $t_1$ we have some segment $S$. Then the tail of $S$ is $S_1$, and the head is $S_{|S|}$. $S(t)$ is the segment whose tail is $S_1$ at time $t$, hence $S(t_1) = S$. Finally, $S(t)_{|S(t)|}$ is the head of the segment $S(t)$. 

In the $k>1$ setting, locusts can frequently move between tracks, which complicates our study of $T_{stable}$. Crucially, however, the number of segments on any individual track is non-increasing. This is because, first, as shown in the previous section, locusts moving and conflicting on the same track can never create new segments. Second, by the locust model, locusts can only move into another track when this places them between two locusts that already belong to some (clockwise or counterclockwise) segment. 

That being said, locusts moving in and out of a given track makes the technique we used in the previous section unfeasible. In the following definitions of \textit{compact} and \textit{deadlocked} locust sets, our goal is to identify configurations of locusts on a given track which locusts cannot enter from another track. Such configurations can be studied locally, focusing only on the track they are in. In the next several lemmas, we will bound the amount of time that can pass without either the number of segments decreasing, or all segments entering into deadlock.

\begin{definition}
We call a sequence of locusts  $X_1, X_2, \ldots$ \textbf{compact} if $X_{i+1} = X_i^{\rightarrow}$ and either:

\begin{enumerate}
    \item every locust in $X$ has clockwise heading and for every $i < |X|$, $dist^{c}(X_i,X_{i+1})\leq 2$, or
    
    \item every locust in $X$ has counterclockwise heading and for every $i < |X|$, $dist^{cc}(X_i,X_{i+1})\leq 2$.
\end{enumerate}

An unordered set of locusts is called compact if there exists an ordering of all its locusts that forms a compact sequence.
\end{definition}

\begin{definition}
Let $X = \{X_1, X_2, \ldots X_j\}$ and $Y = \{Y_1, Y_2, \ldots Y_k\}$ be two compact sets, such that the locusts of $X$ have clockwise heading and the locusts of $Y$ have counterclockwise heading. $X$ and $Y$ are \textbf{in deadlock} if $dist^c(X_j,Y_k) = 1$. (See Figure \ref{deadlockedfigure})
\end{definition}

A compact set of locusts $X$ is essentially a platoon of locusts all on the same track which are heading in one direction, and are all jammed together with at most one empty space between each consecutive pair. As long as $X$ remains compact, no new locusts can enter the track between any two locusts of $X$, because the model states that locusts do not move vertically into empty locations to which a locust is attempting to move horizontally, and the locusts in a compact set are always attempting to move horizontally to the empty location in front of them.

\begin{figure}[!ht]
\centerline{\includegraphics[scale=.35]{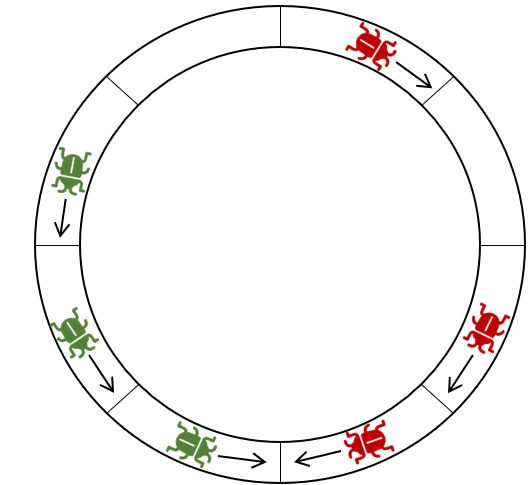}}
\caption{Two segments in deadlock, colored green and red.}
\label{deadlockedfigure}
\vspace{-6mm}
\end{figure}

\begin{definition}
A maximal compact set is a set $X$ such that for any locust $A \notin X$, $X \cup {A}$ is not compact.
\end{definition}

A straightforward observation is that locusts can only belong to one maximal compact set:

\begin{observation}
\label{maximalcompactsetobservation}
Let $A$ be a locust. If $X$ and $Y$ are maximal compact sets containing $A$, then $X = Y$.
\end{observation}

\begin{Lemma}
\label{deadlocksegmentstayslemma}
Let $X$ and $Y$ be two sets of locusts in deadlock at the beginning of time $t$. Then at every subsequent time step, the locusts in $X \cup Y$ can be separated into sets $X'$ and $Y'$ that are in deadlock, or the locusts in $X \cup Y$ all have identical heading.
\end{Lemma}

\begin{proof}
Let $X = \{X_1, X_2, \ldots X_j\}$ and $Y = \{Y_1, Y_2, \ldots Y_k\}$ be compact sets such that $X_{i+1} = X_i^{\rightarrow}$, $Y_{i+1} = Y_i^{\rightarrow}$. It suffices to show that if $X$ and $Y$ are in deadlock at time $t$, they will remain that way at time $t+1$, unless $X \cup Y$'s locusts all have identical heading. Let us assume without loss of generality (``w.l.o.g.'') that $X$ has clockwise heading, and therefore $Y$ has counterclockwise heading. By the definition of deadlock, at time $t$, $X_j$ and $Y_k$ conflict, and the locust that loses joins the other set. Suppose w.l.o.g. that $X_j$ is the locust that lost. If $|X| = 1$, then the locusts all have identical heading, and we are done. Otherwise, set  $X' = \{X_1, \ldots X_{j-1}\}$ and $Y' = \{Y_1, Y_2, \ldots Y_k, X_j\}$. Note that since $X$ and $Y$ are compact at time $t$, no locust could have moved vertically into the empty spaces between pairs of locusts in $X \cup Y$. Furthermore the locusts of $X$ and $Y$ all march towards $X_j$ and $Y_k$ respectively, hence the distance between any consecutive pair $X_i, X_{i+1}$ or $Y_i, Y_{i+1}$ could not have increased. Thus $X'$ and $Y'$ are compact.

To show that $X'$ and $Y'$ are deadlocked at time $t+1$, we need just to show that $dist^c(X_{j-1}, X_j)$ is $1$ at time $t+1$. Since the distances do not increase, if $dist^c(X_{j-1}, X_j)$ was $1$ at time $t$, we are done. Otherwise $dist^c(X_{j-1}, X_j) = 2$ at time $t$, and since $X_j$ did not move (it was in a conflict with $Y_k$), $X_{j-1}$ decreased the distance in the last time step, hence it is now $1$.  
\qed\end{proof}

\begin{Lemma}
\label{deadlockinoneringlemma}
Suppose $P$ and $Q$ are the only segments on track $\mathcal{K}$ at time $t_0$, and $P$'s locusts have clockwise heading. Let $d = dist^c(P_1, Q_1)$. After at most 3d time steps, $P(t_0+3d)$ and $Q(t_0+3d)$ are in deadlock, or the track is stable.
\end{Lemma}

\begin{proof}
The track  $\mathcal{K}$ consists of locations of the form $(x,y)$ for some fixed $y$ and $1 \leq x \leq n$. For brevity, in this proof we will denote the location $(x,y)$ simply by its horizontal coordinate, i.e., $x$, by writing $(x) = (x,y)$. 

We may assume w.l.o.g. that $t_0 = 0$, and that $P_1$ is initially at $(0)$. Note that this means $Q_1$ is at $(d)$ at time $0$. If at any time $t \leq 3d$, the track is stable, then we are done, so we assume for contradiction that this is not the case. This means that $P_1$ and $Q_1$ do not change their headings before time 3d. This being the case, we get that $dist^c(P_1, Q_1)$ is non-increasing before time 3d. As the segments $P(t)$ and $Q(t)$ move towards each other at every time step $t \leq 3d$, we can consider only the interval of locations $[0,d]$, i.e., the locations $(0), (1), \ldots (d)$. We then define the distance $dist(\cdot, \cdot)$ between two locusts in this interval whose $x$-coordinates are $x_1$ and $x_2$ as $|x_1 - x_2|$.

At any time $t \leq 3d$, we may partition the locusts in $[0,d]$ into maximal compact sets of locusts. This partition is unique, by Observation \ref{maximalcompactsetobservation}. Let us label the maximal compact sets of locusts that belong to $P(t)$ as $\mathcal{C}_1^t, \mathcal{C}_2^t, \ldots \mathcal{C}_{c_t}^t$, where the segments are indexed from $1$ to $c_t$, sorted by increasing $x$ coordinates, such that $\mathcal{C}_1^t$ contains the locusts closest to $(0)$. Analogously, we label the maximal compact sets that belong to $Q(t)$ as $\mathcal{W}_1^t, \mathcal{W}_2^t, \ldots \mathcal{W}_{w_t}^t$, with indices running from $1$ to $w_t$, sorted by decreasing $x$-coordinates such that $\mathcal{W}_1^t$ contains the locusts that are closest to $(d)$ (see Figure \ref{lemmaconstructfig}). In this proof, the distance between two sets of locusts $X, Y$, denoted $dist(X,Y)$, is defined simply as the minimal distance between two locusts $A \in X, B \in Y$. Our proof will utilise the functions:

\begin{equation} 
\begin{split}
L_1(t) = \sum_{i=1}^{c_t-1} dist(\mathcal{C}_i^t,\mathcal{C}_{i+1}^t), &\ 
L_2(t) = \sum_{i=1}^{w_t-1} dist(\mathcal{W}_i^t, \mathcal{W}_{i+1}^t) \\
L_3(t) = dist(\mathcal{C}_{c_t}^t, \mathcal{W}_{w_t}^t), & \ 
L(t) = L_1(t) + L_2(t) + L_3(t) \\
\end{split}
\end{equation}

$L_1(t)$ is the sum of distances between consecutive clockwise-facing sets in the partition at time $t$. $L_2(t)$ is the sum of distances between the counterclockwise sets. $L_3(t)$ is the distance between the two closest clockwise and counterclockwise facing sets. The function $L(t)$ is the sum of distances between consecutive compact sets in the partition. When $L(t) = 1$, there are necessarily only one clockwise and one counterclockwise facing sets in the partition, which must equal $P(t)$ and $Q(t)$ respectively. Furthermore, $L(t) = 1$ implies that the distance between $P(t)$ and $Q(t)$ is $1$. Hence when $L(t) = 1$, $P(t)$ and $Q(t)$ are both in deadlock. The converse is true as well, hence $L(t) = 1$ if and only if $P(t), Q(t)$ are in deadlock. We will use $L(t)$ as a potential or ``Lyapunov'' function \cite{la2012lyapunov} and show it must decrease to $1$ within 3d time steps. By Lemma $\ref{deadlocksegmentstayslemma}$, once $P$ and $Q$ are in deadlock they will remain in deadlock until one of them is eliminated, which completes the proof.

\begin{figure}[!ht]
\centerline{\includegraphics[scale=.6]{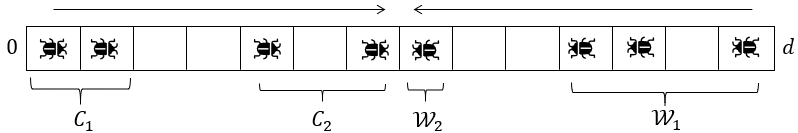}}
\caption{A partition into maximal compact subsets as in our construction. In this configuration, $L_1(t) = 3$, $L_2(t) = 3$, $L_3(t) = 1$, and $L(t) = 7$.}
\label{lemmaconstructfig}
\vspace{-6mm}
\end{figure}

Let us denote by $max(X)$ the locust with maximum $x$-coordinate in $X$, and by $min(X)$ the locust with minimal $x$-coordinate. We may also use $max(X)$ and $min(X)$ to denote the $x$ coordinate of said locust. Note that  $dist(\mathcal{C}_{i}^t,\mathcal{C}_{i+1}^t)$ is the distance between $max(\mathcal{C}_{i}^t)$ and $min(\mathcal{C}_{i+1}^t)$. 

Recall that in the locust model, every time step is divided into a phase where locusts move horizontally (on their respective tracks), and a phase where they move vertically. First, let us show that the sum of distances $L_1(t)$ does not increase due to changes in either the horizontal or vertical phase. Since $L_1(t)$ is the sum of distances between compact partition sets whose locusts move clockwise, and for all  $\mathcal{C}_{i}^t$ except perhaps $\mathcal{C}_{c_t}^t$, $max(\mathcal{C}_{i}^t)$ always moves clockwise, the distance $dist(\mathcal{C}_{i}^t,\mathcal{C}_{i+1}^t)$ does not increase as a result of locust movements (note that clockwise movements of  $max(\mathcal{C}_{i}^t)$ do not result in a new compact set because the rest of the locusts in $C_i^t$ follow it). Furthermore, since conflicts cannot result in a new maximal compact set in the partition, conflicts do not increase $L_1(t)$. Hence, $L_1(t)$ does not increase in the horizontal phase. In the vertical phase, clockwise-heading locusts entering the track either create a new set in the partition, which does not affect the sum of distances (as they then merely form a ``mid-point'' between two other maximal compact sets), or they join an existing compact set, which can never increase $L_1(t)$. By the locust model, the only locusts that can move tracks are  $max(\mathcal{C}_{c_t}^t)$ and  $min(\mathcal{W}_{w_t}^t)$, since these are the only locusts for which the condition $b(A) \neq b(A^{\rightarrow})$ is true, so locusts moving tracks cannot increase $L_1(t)$ either. In conclusion, $L_1(t)$ is non-increasing at any time step. By analogy, $L_2(t)$ is non-increasing.

Similar to $L_1$ and $L_2$, the distance $L_3(t)$ cannot increase as a result of locusts entering the track. It can increase as a result of a locust conflict which eliminates either $\mathcal{W}_{w_t}^t$ or $\mathcal{C}_{c_t}^t$, but such an increase is compensated for by a comparable decrease in either $L_1(t)$ or $L_2(t)$. It is also simple to check that, since $P(t)$ and $Q(t)$ are always moving towards each other when they are not in deadlock (i.e., when $L(t) > 1$),  there will be at least two compact sets in the partition that decrease their distance to each other, hence $L_1, L_2$ or $L_3$ must decrease by at least $1$ in the horizontal phase. 

To conclude: $L_1(t)$ and $L_2(t)$ are non-increasing. $L_3(t)$ is non-increasing during the horizontal phase and as a result of new locusts entering $\mathcal{K}$. If $L(t) > 1$, $L(t)$ decreases during each horizontal phase. Hence, $L(t)$ decreases in every time step where $L(t) > 1$ and no locusts in $\mathcal{K}$ move to another track.

What happens when locusts in $\mathcal{K}$ do move to another track? As proven, $L_1(t)$ and $L_2(t)$ do not increase. However, the distance $L_3(t)$ will increase, since the only locusts that can move tracks are $max(\mathcal{C}_{w_t}^t)$ and $min(\mathcal{W}_{c_t}^t)$. It is straightforward to check that when $\mathcal{C}_{c_t}^t$ contains more than one locust, $L_3(t)$ will increase by at most $2$ as a result of $max(\mathcal{C}_{w_t}^t)$ moving tracks. When $\mathcal{C}_{c_t}^t$ contains exactly one locust, $L_3(t)$ can increase significantly (as $L_3(t)$ then becomes the distance between $C_{c_t-1}^t$ and $W_{w_t}$), but any increase is matched by the decrease in $L_1(t)$ as a result of $\mathcal{C}_{c_t}^t$ being eliminated. Analogous statement hold for $\mathcal{W}_{w_t}^t$, and hence $L_3(t)$ can increase by at most $2$ as a result of one locust moving out of the track. We need to bound, then, the number of locusts in $\mathcal{K}$ that move tracks before time $3d$. We define the potential function $F(t)$:

\begin{equation}
\begin{split}
    F(t) =  \sum_{i=1}^{c_t-1} (dist(\mathcal{C}_i^t,\mathcal{C}_{i+1}^t) - 1) +  \sum_{i=1}^{w_t-1} (dist(\mathcal{W}_i^t, \mathcal{W}_{i+1}^t) - 1) + |P(t) \cup Q(t)| = \\
    = L_1(t) + L_2(t) - c_t - w_t + |P(t) \cup Q(t)| 
\end{split}
\end{equation}

$F(t)$ is the sum of the empty locations between consecutive compact sets in the partition whose locusts have the same heading, plus the number of locusts in $\mathcal{K}$. Note that $F(t) \geq 0$ at all times $t$. We will show $F(t)$ is non-increasing, and that it decreases whenever a locust leaves the track. Hence, at most $F(0)$ locusts can leave the track.

Let us show that $F(t)$ is non-increasing. We already know $L_1$ and $L_2$ are non-increasing. In the horizontal phase, $|P(t) \cup Q(t)|$ is of course unaffected. $c_t$ and $w_t$ can decrease as a result of maximal compact sets merging, hence increasing $F$, but this can only happen when the distance between two such sets has decreased, hence the resulting increase to $F$ is undone by a decrease in $L_1$ and $L_2$. Hence, $F(t)$ does not increase because of locusts' actions during the horizontal phase.

Likewise, locusts leaving $\mathcal{K}$ can decrease $c_t$ or $w_t$ when they cause a maximal compact set to be eliminated, but this is matched by a comparable decrease in $L_1$ or $L_2$ which means that $F$ does not increase due to locusts moving out of the track. Furthermore,  $|P(t) \cup Q(t)|$ decreases when this happens. Hence, a locust moving out of the track decreases $F(t)$ by at least $1$. Finally, let us show that locusts entering the track does not increase $F(t)$.

At time $t$, locusts can only enter the track at empty locations that are found in intervals of the form $[max(\mathcal{W}_i^t), min(\mathcal{W}_{i+1}^t)]$ or  $[max(\mathcal{C}_{i+1}^t), min(\mathcal{C}_{i}^t)]$ for some $i$. In particular, locusts cannot enter empty locations that are between two locusts belonging to the same compact set (because a locust in that set will always be attempting to move to that location in the next time step, and the model disallows vertical movements to such locations), nor can they enter the track on the empty locations between $min(\mathcal{C}_{c_t}^t)$  and $max(\mathcal{W}_{w_t}^t)$. Thus, locusts entering the track at time t decrease the amount of empty locations between two clockwise or
counterclockwise compact partition sets (and perhaps cause the sets between which they enter to merge into a single compact set). This will always decrease $L_1(t) + L_2(t) - c_t - w_t$ by at least $1$ and increase $|P(t) \cup Q(t)|$ by $1$. On net, we see that new locusts entering $\mathcal{K}$ either decreases or does not affect $F$. 

In conclusion, $F(t)$ is non-increasing, and any time a locust moves to another track, $F(t)$ decreases by $1$. Thus, at most $F(0)$ locusts can move from $\mathcal{K}$ to another track. Recall that locusts moving out of the track can increase $L(t)$ by at most $2$. Hence after at most $L(0) + 2F(0) \leq d + 2d = 3d$ time steps, $L(t) = 1$.
\qed\end{proof}

\begin{Lemma}
\label{timetodeadlockalllemma}
Let $seg(t)$ denote the set of segments in all tracks at time $t$. At time $t+3n$, either every segment is in deadlock with some other segment, or $|seg(t+3n)| < |seg(t)|$.
\end{Lemma}

\begin{proof}
Consider some track $\mathcal{K}$ and a segment $P$ which is in that track at time $t$. Let us assume that $|seg(t+3n)| = |seg(t)|$, and show that $P(t+3n)$ must be in deadlock with another segment. At any time $t' \geq t$, as long as the number of segments on $\mathcal{K}$ does not decrease, the locusts of $P(t')$ will be marching towards locusts of another segment, which we will label $Q(t')$. They cannot collide or conflict with locusts belonging to any segment other than $Q(t')$. Hence, other segments in $\mathcal{K}$ do not affect the evolution of $P(t)$ and $Q(t)$ before time $t+3n$, and we can assume w.l.o.g. that $P(t)$ and $Q(t)$ are the only segments in $\mathcal{K}$ at time $t$. Let $d$ be as in the statement of Lemma \ref{deadlockinoneringlemma}. Since $n \geq d$, Lemma \ref{deadlockinoneringlemma} tells us that at some time $t \leq t^* \leq t+3n$, $P(t^*)$ and $Q(t^*)$ must be in deadlock. Since by Lemma \ref{deadlocksegmentstayslemma}, $P$ and $Q$ must remain in deadlock until one of them is eliminated, we see that at time $t+3n$ they must still be in deadlock, since we assumed $|seg(t)| = |seg(t+3n)|$.
\qed\end{proof}

\begin{Theorem}
$\mathbb{E}[T_{stable}] \leq  \frac{3}{4}mn + \frac{\pi^2}{24}m^2 = \mathcal{O}(mn + m^2)$
\label{mnm2bound}
\end{Theorem}

\begin{proof}
Let $|seg(t)|$ denote the number of segments at time $t$. $\mathbb{E}[T_{stable}]$ can be computed as the sum of times $\mathbb{E}[T_2 + T_4 + \ldots + T_{|seg(0)|}]$, where $T_i$ is the expected time until the number of segments drops below $i$, if it is currently $i$ (we increment the index by $2$ since segments are necessarily eliminated in pairs). 

Let us estimate $E[T_{2i}]$. Suppose that at time $t$, the number of segments is $2i$. Then after $3n$ steps at most, either the number of segments has decreased, or all segments are in deadlock. There are in total $i$ pairs of segments in deadlock, and as there are $m$ locusts, there must be a pair $P(t+3n), Q(t+3n)$ that contains at most $m/i$ locusts. By Lemma \ref{deadlocksegmentstayslemma}, $P(t+3n), Q(t+3n)$ remain in deadlock until either $P$ or $Q$ is eliminated. We can compute precisely how long this takes, since at every time step after time $t+3n$, the heads of $P$ and $Q$ conflict, resulting in one of the segments increasing in size and the other decreasing. Hence, the expected time it takes $P$ or $Q$ to be eliminated is precisely the expected time it takes a symmetric random walk starting at $0$ to reach either $|P(t+3n)|$ or $-|Q(t+3n)|$, which is $|P(t+3n)| \cdot |Q(t+3n)| \leq (\frac{m}{2i})^2$. Hence, $E[T_{2i}] \leq 3n +  (\frac{m}{2i})^2$. Consequently:

\begin{equation}
     \mathbb{E}[T_2 + T_4 + \ldots + T_{|seg(0)|}] \leq 3n \cdot \frac{|seg(0)|}{2} + \sum_{i=1}^{\infty}  (\frac{m}{2i})^2 \leq \frac{3}{2}mn + \frac{\pi^2}{24}m^2
\end{equation}

Where we used the inequality $|seg(0)| \leq m$ and the identity $\sum_{i=1}^{\infty}  (\frac{1}{i})^2 = \frac{\pi^2}{6}$.
\qed\end{proof}

We prove next that $E[T_{stable}]=\mathcal{O}(\log (k)n^2)$. For this, we require the following result:

\begin{Lemma}
Consider $k$ independent random walks with absorbing barriers at $0$ and $2n$, i.e., random walks that end once they reach $0$ or $2n$. The expected time until \textbf{all} $k$ walks end is $\mathcal{O}(n^2 \log (k))$.
\label{randomwalklimitlemmaksqrt}
\end{Lemma}

\begin{proof}
First, let us set $k=1$ and estimate the probability that the one walk has not ended by time $t$. Let $P$ be the transition probability matrix of the random walk, and let $\textbf{v}$ be the vector describing the initial probability distribution of the location of the random walker. Then $\textbf{v} P^t$ is the probability distribution of its location after $t$ time steps \cite{markovref}. The evolution of $\textbf{v} P^t$ is well-studied and relates to ``the discrete heat equation'' \cite{lawler2010randomwalkheatequation}. The probability that the walk has not ended at time $t$ is the sum $\sum_{i=1}^{2n-1} \textbf{v}(i)$.   Asymptotically, this sum is bounded by  $\mathcal{O}(\lambda^t)$ where $\lambda = cos(\frac{\pi}{2n})$ is the $2$nd largest eigenvalue of $P$ (cf. \cite{lawler2010randomwalkheatequation}). 

Returning to general $k$, let $T_k$ be a random variable denoting the time when all $k$ walks end. By looking at the series expansion of $cos(1/x)$, we may verify that for  $n > 1$, $cos(\frac{\pi}{2n}) < 1 - \frac{1}{n^2}$. From the previous paragraph, and because the walks are independent, we therefore see that 

\begin{equation}
        Pr(\mathcal{T}_k \geq t)  = 1 - Pr(\mathcal{T}_1 < t)^k = 1 - \big(1 - \mathcal{O}( \lambda^t)\big)^k = 1 - \big(1 - \mathcal{O}((1 - \frac{1}{n^2})^t)\big)^k 
\end{equation}

Consequently, for $t \gg n^2$, the following asymptotics hold for some constant $C$: 

\begin{equation}
    Pr(\mathcal{T}_k \geq t) < 1 - (1 - Ce^{-t/n^2})^k   
\end{equation}

Where we used the fact that $(1+x/n)^n \to e^x$ as $n \to \infty$. Note that $Pr(\mathcal{T}_k \geq t+n^2\log(C)) < 1 - (1 - e^{-t/n^2})^k$. Hence:

\begin{equation}
    \begin{split}
        \mathbb{E}[\mathcal{T}_k] = \int_{0}^{\infty} Pr(T_k > t) dt   \leq n^2\log(C) + \int_{0}^{\infty} 1 - (1 - e^{-t/n^2})^k  dt = \\ = n^2\log(C)  + \int_{0}^{\infty} 1 - \sum_{j=0}^k \binom{k}{j} (-1)^j  e^{-tj/n^2} dt = n^2\log(C)  + - \sum_{j=1}^k \binom{k}{j} (-1)^j  \int_{0}^{\infty} e^{-tj/n^2} dt = \\ =n^2\log(C)  + -n^2 \sum_{j=1}^k \binom{k}{j}  \frac{(-1)^j }{j} = \mathcal{O}(n^2\log (k)) 
    \end{split}
\end{equation}

Where we used the equality $\sum_{j=1}^k \binom{k}{j}  \frac{(-1)^j }{j} = -H_k$,  $H_k$ being the $k$th harmonic number.\qed\end{proof}

\begin{Theorem}
$\mathbb{E}[T_{stable}] =  \mathcal{O}(\log (k) \cdot  n^2)$
\label{kn2bound}
\end{Theorem}

\begin{proof}
Let $seg_i(t)$ denote the number of segments in track $i$ at time $t$, and define $\mathcal{M}_t = \max_{1\leq i \leq k} seg_i(t)$. Let us bound the expected time it takes for $\mathcal{M}_t$ to decrease. Define the set $K(t)$ to be all tracks that have  $|\mathcal{M}_t|$ segments at time $t$. Then $\mathcal{M}_t$ decreases at the first time $t' > t$ when all tracks in $K(t)$ have had their number of segments decrease. We may bound this with the following argument: slightly generalizing Lemma \ref{timetodeadlockalllemma} to hold for subsets of tracks\footnote{Lemma \ref{timetodeadlockalllemma} holds not just for the set $seg(t)$ but for the segments in a given subset of tracks, with the proof being virtually identical. Here we apply the Lemma to the subset $K(t+3n)$.} , if $\mathcal{M}_t$ doesn't decrease after $3n$ time steps (i.e., $\mathcal{M}_t=\mathcal{M}_{t+3n}$), all tracks in $K(t+3n)$ now have all their segments in deadlock. The number of deadlocked segment pairs at every track in $K(t+3n)$ is $\mathcal{M}_t/2$, so in every such track there is such a pair with at most $2n/\mathcal{M}_t$ locusts. By Lemma \ref{randomwalklimitlemmaksqrt}, using similar reasoning to Theorem \ref{mnm2bound}, these pairs of deadlocked segments resolve into a single segment after at most $c \cdot \log (k)\big(\frac{2n}{M_t}\big)^2$ expected time for some constant $c$. Hence, the number of expected time steps for $\mathcal{M}_t$ to decrease is bounded above by $3n + c\log (k)\big(\frac{2n}{M_t}\big)^2$.

$T_{stable}$ is the first time when $\mathcal{M}_t = 0$. Let us assume $n$ is even for simplicity (the computation will hold regardless, up to rounding). We have that $\mathcal{M}_0 \leq n$, and $\mathcal{M}_t$ decreases in leaps of $2$ or more (since segments can only be eliminated in pairs). Hence, $T_{stable}$ is bounded by the amount of time it takes $\mathcal{M}_t$ to decrease at most $n/2$ times. By linearity of expectation and the previous paragraph, this can be bounded by summing $3n + c\log (k)\big(\frac{2n}{M_t}\big)^2$ over $M_t = n, n-2, n-4, \ldots 2$:

\begin{equation}
    \begin{split}
        \mathbb{E}[T_{stable}] \leq  \frac{n}{2} \cdot 3n +  c\log (k)\big(\frac{2n}{n}\big)^2 + c\log (k)\big(\frac{2n}{n-2}\big)^2 + \ldots + c\log (k)\big(\frac{2n}{2}\big)^2 \leq \\
        \leq \frac{3}{2}n^2 + 4c\log (k)n^2 \sum_{i=1}^{\infty}  (\frac{1}{2i})^2 = \frac{3}{2}n^2 + \frac{\pi^2}{6}c \log (k)n^2 = \mathcal{O}(\log (k)n^2)
    \end{split}
\end{equation} As claimed. \qed\end{proof}

The proof of Theorem \ref{klargebounds} follows immediately from Theorems \ref{mnm2bound} and \ref{kn2bound}, by taking the minimum. \qed

\subsubsection{Erratic track switching and global consensus}

Theorem \ref{klargebounds} shows that, after finite expected time, all locusts on a track have identical heading. This is a stable \textit{local} consensus, in the sense that two different tracks may have locusts marching in opposite directions forever. We might ask what modifications to the model would force a \textit{global} consensus, i.e., make it so that stabilization occurs only when all locusts across \textit{all} tracks have identical heading. There is in fact a simple change that would force this to occur: let us assume that at time step $t$ any locust has some probability of acting ``eratically'' in either the vertical or horizontal phases:

\begin{enumerate}
    \item With probability $r$, a locust might behave erratically in the horizontal phase, staying in place instead of attempting to move according to its heading.
    \item With probability $p$, a locust may behave erratically in the vertical phase, meaning that even if the vertical movement  conditions (1)-(3) of the model (see Section \ref{modelsection}) are not fulfilled, the locust attempts to move vertically to an adjacent empty space on the track above or below them (if such empty space exists).
\end{enumerate}

These behaviours are independent, and so a locust may behave erratically in both the vertical and horizontal phases, in just one of them, or in neither.

The next theorem shows that the existence of erratic behaviour forces a global consensus of locust headings. The goal is to prove that there is some finite time after which all locusts must have the same heading. Note that the bound we find for this time is crude, and is not intended to approximate $T_{stable}$. We study the question of how $p$ affects $T_{stable}$ empirically in the next section.

\begin{Theorem}
Assuming there is at least one empty space (i.e., $m < nk$), and the probability of erratic track switching is $0 < r,p < 1$, the locusts all have identical heading in finite expected time.
\end{Theorem}

\begin{proof}
Our goal is to show that all locusts must have identical heading in finite expected time. We will find a crude upper bound for this time. It suffices to show that as long as there are two locusts with different headings in the system (perhaps not on the same track), there is a bounded-above-0 probability $q$ that within a some constant, finite number of time steps $C$ (we will show $C= \mathcal{O}(\log (k)n^2+nk)$), the number of locusts with clockwise heading will increase. This amounts to showing that there is a sequence of events, each individual event happening with non-zero probability, that culminates in a conflict between two locusts occurring (since any conflict has probability $0.5$ of increasing the number of clockwise locusts). Since $q > 0$, the only stable state of locust headings is the state where all locusts have identical heading, as otherwise there is always some probability that all locusts will have clockwise heading after $m \cdot C$ time steps; this completes the proof.

Let us show such a sequence of events. First let us consider the case where there is a track in which two locusts have non-identical headings. In this case, assuming no locusts behave erratically for $\mathcal{O}(\log (k)n^2)$ steps (which occurs with a tiny but bounded-above-0 probability since $p,r > 0$), Theorem \ref{klargebounds} tells us that in expected $\mathcal{O}(\log (k)n^2)$ steps, locusts on the same track will have identical heading. Hence, there is a sequence of events that happens with non-zero probability which leads to local consensus in the tracks. 

If any conflict occurs during this sequence, we are done. Otherwise, we need to show a sequence of events that leads to a conflict, assuming all tracks are stable. The only thing that causes locusts in local consensus to move tracks is erratic behaviour. If two adjacent tracks have locusts with non-identical heading, and there is at least one empty space in one of them, then (since $r > 0$) with some probability within at most $n$ time steps an empty space in one track will be vertically adjacent to a locust in the other track. At this point, with probability $p$, that locust will move from one track to the other. This creates a situation where in one track there are locusts of different headings again. If the erratic locust moves tracks at the right time, upon moving it will be adjacent to another locust in its new track, whose heading is different. Hence, the erratic locust will enter a conflict in the next time step, which will increase the number of clockwise locusts with probability $0.5$.

Now let us consider a pair of two adjacent tracks with locusts of different headings such that there no empty space in one of them. We note that since there is at least one empty location in \textit{some} track, erratic behaviour can cause that empty location to move vertically in an arbitrary fashion until, after at most $k$ movements, it enters a track from the pair. With non-zero probability, this can take at most $nk$ time steps, after which we are reduced to the situation in the previous paragraph.

A pair of adjacent tracks that have locusts with different headings must exist unless there is global consensus. Hence, in every $\mathcal{O}(\log (k)n^2+nk)$ time steps where there is no global consensus, there is a some probability $q > 0$ that the number of clockwise-heading locusts will increase. \qed\end{proof}

\section{Simulation and empirical evaluation}

\begin{figure}[!ht]
\centerline{\includegraphics[scale=.6]{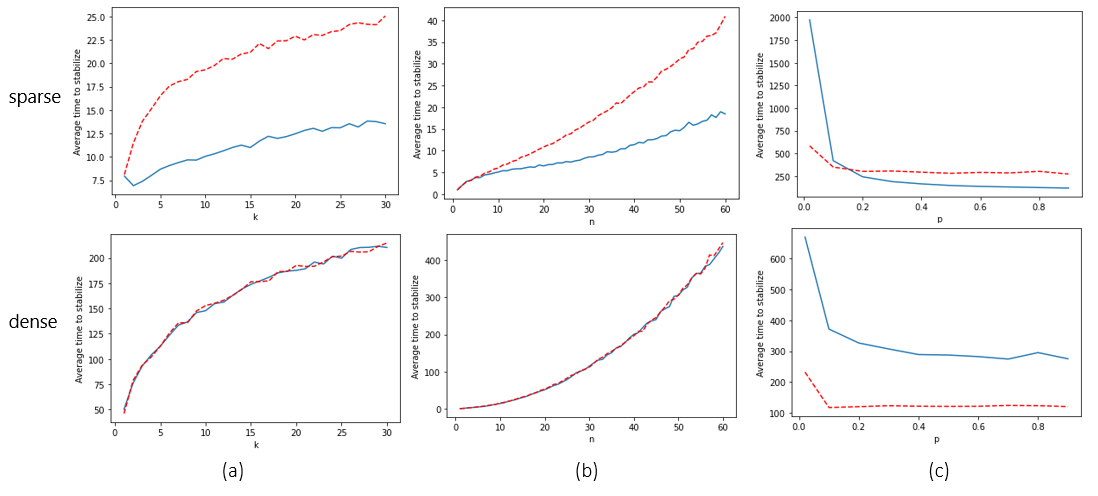}}
\caption{Simulations of the locust model. The $y$ axis is $T_{stable}$. Column (a) measures $T_{stable}$ for $k=1...30$, with $n$ fixed at 30. Column (b) measures $T_{stable}$ for $n=1...60$, with $k$ fixed at 5. Column (c) measures $T_{stable}$ with $n=30, k=5$, and $p$ (the probability of erratic behaviour) going from $0$ to $1$. The top row measures $T_{stable}$ for a sparse locust configuration, while the bottom row does so for a dense configuration. The dashed line estimates $T_{stable}$ when locusts never switch tracks (except while behaving erratically in column c); the blue line estimates $T_{stable}$ when locusts switch tracks as often as the model rules allow.}
\label{graphdata}
\vspace{-6mm}
\end{figure}

Let us explore some questions about the expected value of $T_{stable}$ through numerical simulations. Certain aspects of the locusts' dynamics were not studied in our formal analysis: the most interesting of which is the helpful effects of track switching on $T_{stable}$. Recall that our model allows locusts to switch tracks if this would enable them to avoid a conflict and join a track where \textit{locally}, locusts are marching in their same direction. At least in principle, this seems like it should help our locusts achieve local stability faster, hence decrease $T_{stable}$. However, recall also that we do not specify \textit{when} locusts switch tracks, which means that some locusts might never switch tracks, or they might choose to do so in the worst possible moments. Hence, the positive effect track-switching usually has on $T_{stable}$ cannot be reflected in the bounds we found for $\mathbb{E}[T_{stable}]$, since these bounds must reflect all possible locust behaviours. Under ordinary circumstances, however, it seems as though frequent track switching should noticeably decrease the time to local stabilization. As we shall see numerically, this is indeed the case. This justifies the track-switching behaviour as a mechanism that, despite being highly local, enables the locusts to come to local consensus about the direction of motion sooner.

In Figure \ref{graphdata}, (a) and (b), we measure $T_{stable}$ as it varies with $n$ and $k$, assuming the probabilities of erratic behaviour are $0$ (i.e., $r=p=0$). We simulate two different locust configurations: a ``dense'' configuration, and a ``sparse'' configuration. In the dense configuration, $50\%$ of locations are initiated with a locust, with the locations chosen at random. In the sparse configuration, $10\%$ of locations are initiated with a locust (or slightly more, to guarantee all tracks start with $2$ locusts). The locusts are initiated with random heading. We measure the effect of track switching on $T_{stable}$: the opaque lines measure $T_{stable}$ when locusts switch tracks as often as they can (while still obeying the rules of the model), and the dotted lines measure $T_{stable}$ when locusts never switch tracks. For every value of $n$, $k$, we ran the simulation 1000 to 3000 times and averaged $T_{stable}$ over all simulations.

As we can see, in the sparse configuration, track-switching has a significantly positive effect on time  to stabilization. For example, with $k=30$, $n=30$, $T_{stable}$ is approximately $13.5$ when locusts switch tracks as soon as they can, and approximately $25$ when they never switch tracks--nearly double. In the dense configuration, we see that enabling locusts to move tracks has little to no effect, since the locust model rarely allows them to do so due to the tracks being overcrowded.

In column (c) of Figure \ref{graphdata}, we measure how a non-zero probability $p$ of erratic behaviour affects $T_{stable}$. We set $r = 0$. As we proved in the previous section, whenever $p > 0$, stabilization requires \textit{global} rather than local consensus. Hence, we cannot directly compare the $T_{stable}$ of these graphs with columns (a) and (b), where $T_{stable}$ measures the time to local consensus. We see that $\mathbb{E}[T_{stable}]$ approaches $\infty$ as $p$ goes to $0$, as one would expect, since when $p = 0$, global stability can never occur in some initial configurations. $\mathbb{E}[T_{stable}]$ decreases sharply as $p$ goes to some critical point around $0.1$, and decreases at a slower rate afterwards. It is interesting to note that low probability of erratic behaviour affects  $\mathbb{E}[T_{stable}]$ significantly more in the \textit{sparse} configuration, where for $p = 0.02$, if locusts also switch tracks whenever the model allows them, $\mathbb{E}[T_{stable}]$ was measured as being approximately $1974$, as opposed to $669$ in the dense configuration. One of the core reasons for this seems to be that, in the sparse configuration, when a locust erratically moves to a track with a lot of locusts not sharing its heading, it will often be able to \textit{non-erratically} move back to its former track, thus preventing locust interactions between tracks of different headings. When we disabled the locusts' ability to switch tracks non-erratically, $T_{stable}$ was significantly smaller in the sparse configuration ($\mathbb{E}[T_{stable}] \approx 232$ for $p = 0.02$). 

Based on the above, we make the curious observation that, while non-erratic track switching accelerates local consensus, for some track-switching behaviours, it will in fact decelerate the attainment of global consensus. This is seen by the fact that frequent non-erratic  track-switching was helpful in Columns (a) and (b) of Figure  \ref{graphdata}, but increased time to stabilization in Column (c). This is perhaps a very natural observation, because agents that aggressively switch tracks will attempt to avoid conflict as often as possible, whereas conflict is necessary to create global consensus. 

\section{Concluding remarks}

We studied collective motion in a model of discrete locust-inspired swarms, and bounded the expected time to stabilization in terms of the number of agents $m$, the number of tracks $k$, and the length of the tracks $n$. We showed that when the swarm stabilizes, there must be a local consensus about the direction of motion. We also showed that, when the model is extended to allow a small probability of erratic behaviour to perturb the system, global consensus eventually occurs. 

A direct continuation of our work would be to find upper bounds on time to stabilization when there is some probability of erratic behaviour. Furthermore, our empirical simulations suggest several curious phenomena related to erratic behaviour: first, there seems to be a clash between ``erratic'' and non-erratic, ``rational'' track-switching, as when locusts switch tracks non-erratically in order to avoid collisions, this seems to accelerate the attainment of local consensus, but mostly hinder the attainment of global consensus. Second, increasing the probability of erratic track-switching $p$  behaviour was helpful in accelerating global consensus up to a point, but in simulations, its impact seemed to fall off past a small critical value of $p$. In future work, it would be interesting to investigate these aspects of the model.

Although our dynamics model is inspired by experiments on locusts, it can be understood in more abstract terms as a model that describes a situation where many agents that wish to maintain a direction of motion are confined to a small space where they exert pressure on each other. It is natural to ask what kinds of collective dynamics, if any, we should expect when this small space has a different topology; rather than a ringlike arena, we might consider, e.g., a square arena. We believe that rich models of swarm dynamics can be discovered through observing natural organisms exert pressure on each other in such environments. In the introduction, we mentioned points of similarity between our model and models of opinion dynamics. We suspect that these points of similarity will remain in settings with non-ringlike arenas, and might provide a starting point for formally modelling and analysing them. 

\section{Acknowledgements}

This research was partially supported by the Israeli Science Foundation grant no. 2306/18. The authors would like to thank Prof. Amir Ayali (Tel Aviv University) for bringing our attention to the locust experiments and for graciously letting us use the image in Figure \ref{reallocustsfigure},  Prof. Ofer Zeitouni (Weizmann Institute of Science) for helpful discussions, and the anonymous reviewers for constructive comments.

\bibliographystyle{plain}
\bibliography{references}

\end{document}